\documentclass{article}

\usepackage{graphicx}
\usepackage{dcolumn}
\usepackage{bm}

\usepackage[utf8]{inputenc}
\usepackage[T1]{fontenc}
\usepackage{mathptmx}
\usepackage{amsthm}
\usepackage{tikz}
\usetikzlibrary{matrix}
\usepackage{bbold}
\usepackage{tabularx}
\usepackage{amsmath}
\usepackage{amssymb}

\newtheorem{theorem}{Theorem}
\newtheorem{proposizione}{Proposition}
\newtheorem{corollario}{Corollary}
\newtheorem{lemma}{Lemma}
\newtheorem{definizione}{Definition}[section]
\theoremstyle{remark}
\newtheorem{remark}{Remark}
\newtheorem{example}{Example}

\newcommand{\bigslant}[2]{{\raisebox{.5em}{$#1$}\left/\raisebox{-.5em}{$#2$}\right.}}

\begin{document}
		
\title{Some Remarks on the Operator Formalism\\for Nonlocal Poisson Brackets\thanks{This article has been submitted to the \textit{Journal of Mathematical Physics}}} 

\date{\today}
\author{Riccardo Ontani}

\maketitle 

\begin{abstract}
	A common approach to the theory of nonlocal Poisson brackets, seen from the operatorial point of view, has been to keep implicit the sets on which these brackets act. In this paper we aim to explicitly define appropriate functional spaces underlying to the theory of 1 codimensional weakly nonlocal Poisson brackets, motivating the definitions, and to prove the validity in this context of some classical results in the field. We start by introducing the spaces for the local case, which will serve as building tools for those in the nonlocal one. The precise definition and the study of these nonlocal functionals are the core of this work; in particular we work out a characterization of the variational derivative of such objects. We then translate everything to the level of manifolds, defining a global version of the functionals, and introduce the notion nonlocal Poisson brackets in this context. We conclude by applying all the machinery to prove a theorem due to Ferapontov. This last application is the natural conclusion of our discussion and shows that the spaces we introduce are suitable objects to work with when studying topics in this theory.
\end{abstract}
\section*{Introduction}
The theory of Poisson brackets over functional spaces has its roots in the work (Ref.~\cite{DubNov}) of B.A. Dubrovin and S.P. Novikov, in which they studied the conditions for local brackets to be skew symmetric and to satisfy the Jacobi identity. A fundamental result of their work was understanding that these conditions have a differential geometric nature: such Poisson brackets acting on local functionals over a manifold $M$ are related to pseudo-Riemannian structures on the manifold. In (Ref.~\cite{Ferapontov}), E.V. Ferapontov studied the same conditions for weakly nonlocal Poisson brackets (whose name comes from the work (Ref.~\cite{MaltNov}) of A.Ya. Maltsev and S.P. Novikov), finding an even richer bond with Riemannian geometry, that involves a link between the theory of these brackets and the theory of hypersurfaces of Euclidean spaces. In this context, computations were brought on without focusing much on specifying the functional spaces on which the theory was rooted, but with the intention of highlighting the links with differential geometry and the applications to mathematical physics. In the work (Ref.~\cite{MaltNov}) of A.Ya. Maltsev and S.P. Novikov a choice for such spaces is made. Here we aim to show that such definition is in fact well posed and we specify the spaces of functions these operators act on.
It's definitely worth mentioning that there are alternative solutions to the problem of building a formal environment around the theory of these brackets. The most influential one comes from the work, completely based on abstract algebra, of A. De Sole and V.G. Kac (for example see (Ref.~\cite{KacDeSole})). For different computational techniques that can be used to prove the theorem of Ferapontov in this setting, two references are (Ref.~\cite{Lorenzoni}) and (Ref.~\cite{Lorenzoni2}). 
The aim of the present work is to show that it's not \emph{necessary} to pass through this more abstract algebraic formalism in order to study these objects from a rigorous point of view.

\section{FIRST DEFINITIONS: LOCAL FUNCTIONALS}

We first introduce a class of functionals that play a role in the theory of local Poisson brackets \cite{DubNov}. 
\begin{definizione}
	We will denote with $\mathcal{S}_n$ the linear space of functions from $\mathbb{R}$ to $\mathbb{R}^n$ whose components are Schwartz functions.
\end{definizione}
Notice that we have an obvious identification $\mathcal{S}_n \simeq \prod^n \mathcal{S}_1$.
This defines a Fréchet space structure on $\mathcal{S}_n$ through the product metric $\prod^n \mathcal{S}_1$ induced by the usual Fréchet metric on $\mathcal{S}_1$.
 We'll consider the linear integral operator $I : \mathcal{S}_1 \rightarrow \mathbb{R}$ defined as $I[h] := \int_\mathbb{R} h(x) dx$. It is is clearly well defined and bounded.
\begin{definizione}
Let $n, N$ be natural numbers. A function $f : \mathcal{S}_n \rightarrow \mathcal{S}_1$ is said to be a \emph{local density of order N} (or \emph{N-local}) if there exists 
$\phi \in C^\infty(\mathbb{R}^{1 + n \cdot (N+1)})$ such that
\begin{itemize}
\item 
$f(u)(x) := \phi(x, u(x), u^{(1)}(x), ..., u^{(N)}(x)) \quad \quad \forall x \in \mathbb{R}$, $u \in \mathcal{S}_n$.
\item Given a bounded subset $B \subset \mathbb{R}^{n (N+1)}$, for each partial derivative $\psi$ (of any order) of $\phi$ we have
\begin{align*}
\sup_{(x, y) \in {\mathbb{R} \times B}} \left \vert \psi(x, y) \right \vert < +\infty
\end{align*} 
\end{itemize}
For every $i \in \lbrace 0 , ..., N \rbrace$, $j \in \lbrace 1, ..., n \rbrace$ we define ${\partial f}/{\partial u^{(i)}_j} : \mathcal{S}_n \rightarrow C_b^\infty(\mathbb{R}, \mathbb{R}^n)$ such that:
\begin{align*}
&\dfrac{\partial f}{\partial u^{(i)}_j}(v)(x) := \dfrac{\partial \phi}{\partial u^{(i)}_j}(x, v(x), v^\prime(x), ..., v^{(N)}(x))
\end{align*} 
An analogous definition is given for derivatives of higher order.
A functional $F : \mathcal{S}_n \rightarrow \mathbb{R} $ is said to be \emph{N-local} if $F = I \circ f$ where $f$ is an N-local density.
In this case we will write $F \in \mathcal{L}_n$.
\end{definizione}

For local functionals, the Gateaux differential exists and takes a particular well known form:
\begin{theorem}[Euler-Lagrange formula]
Pick $v \in  \mathcal{S}_n$, $F \in \mathcal{L}_n$ and let $f$ be the N-local density associated to F.
Consider the $C^{\infty}(\mathbb{R}, \mathbb{R}^n)$ function defined by
\begin{equation} \label{EL}
\dfrac{\delta F}{\delta v_j}(x) :=
(- \dfrac{d}{d x})^i 
(\dfrac{\partial f}{\partial u_j^{(i)}})(v)(x)
\end{equation}
called \emph{variational derivative} of F in u. F is G-differentiable and
\begin{equation}
\label{eq:1}
d_G F(u)[h] = \int_\mathbb{R} \frac{\delta F}{\delta u}(x) h(x) dx
\end{equation}
holds for each $h \in \mathcal{S}_n$. 
\end{theorem}

Notice that in formula (\ref{EL}), as in the rest of the paper, the Einstein summation convention is used. Explicit sums will be written only in particular situations where confusion is possible.
\begin{example}
	Consider the functional $H \in \mathcal{L}_1$ defined by the composition of $I$ with $h:\mathcal{S}_1 \rightarrow \mathcal{S}_1$ given by $h(u) := u^3 - \frac{1}{2}u u_{xx}$.
	We have that $h$ is a 2-local density:
	Define $\phi \in C^\infty \left(\mathbb{R}^4 \right)$ by $\phi(x, y_0, y_1, y_2):= y_0^3 - \frac{1}{2}y_0 y_2$.
	Then $h(u)(x) = \phi(x, u(x), u_x(x), u_{xx}(x))$, and
	every partial derivative $\psi$ of $\phi$ is of the form
	$\psi(x,y_0, y_1, y_2) = \tilde{\psi}(y_0, y_1, y_2)$ for some smooth $\tilde{\psi}$.
 	Then if $B \subset \mathbb{R}^3$ is bounded we have
	\begin{align*}
		\sup_{(x,y) \in \mathbb{R}\times B} \vert\psi(x,y) \vert = \sup_{y \in B} \vert\tilde{\psi}(y) \vert < \infty
	\end{align*}
	being $\tilde{\psi}$ continuous and $B$ bounded. So the partial derivatives are bounded on $\mathbb{R}\times B$ by constants. This local functional is the Hamiltonian for the KdV equation w.r.t. the Gardner-Zakharov-Faddeev bracket \cite{Gardner}.
\end{example}

\section{ADDING NONLOCALITY:\\ WEAKLY NONLOCAL FUNCTIONALS}
Following the approach of A.Ya. Maltsev and S.P. Novikov in (Ref.~\cite{MaltNov}), we introduce the concept of weakly non local operators. The functionals appearing from now on are examples of pseudo-differental operators. A standard reference to the literature concerning the their theory is the book (Ref ~\cite{Hormander}) of L. H\"{o}rmander. In this section we aim to define the smallest extension of the class of local functionals that is closed under the action of the weakly nonlocal Poisson brackets, which will be defined later in section \ref{Sect3}.
In order to introduce this class of functionals, we consider the linear operator
$d^{-1} : \mathcal{S}_1 \rightarrow C_b^{\infty}(\mathbb{R})$ defined by
\begin{equation}
d^{-1}(f)(x) := \dfrac{1}{2} \left(\int_{-\infty}^x f(z)dz  - \int_{x}^{+\infty} f(z)dz \right)
\end{equation} 
This operator is well defined by convergence of the integrals, due to the basic properties of Schwartz functions. Let's highlight three properties of this object:
\begin{itemize}
\item Let $f \in \mathcal{S}_1$. Then $d^{-1}(f)$ is an antiderivative of $f$. More precisely, it's the antiderivative that at $-\infty$ tends to $-\frac{1}{2}\int_\mathbb{R} f dx$.
\item Let $f, g \in \mathcal{S}_1$. Then $\xi := d^{-1}(f) d^{-1}(g)$ is such that $\xi^\prime = d^{-1}(f) g + f d^{-1}(g)$ and 
\begin{align}\label{secondProperty}
\lim_{x \rightarrow +\infty} \xi(x) - \lim_{y \rightarrow -\infty} \xi(y) = 0
\end{align}
\item Let $f, g \in \mathcal{S}_1$. Then $f \cdot d^{-1}(g) \in \mathcal{S}_1$.
\end{itemize}
The third property allows us to give the following definition.
\begin{definizione}
Consider the linear subspaces of the set of functions $\mathcal{S}_n \rightarrow \mathcal{S}_1$ defined inductively by $\mathcal{D}^0_n := \lbrace\text{local densities } \mathcal{S}_n \rightarrow \mathcal{S}_1\rbrace$ and $\mathcal{D}^m_n := span_\mathbb{R} \tilde{\mathcal{D}}^m_n$ where
\begin{align*}
\tilde{\mathcal{D}}^m_n := \left\lbrace g \prod_{\alpha = 1}^A d^{-1}(h_\alpha) \quad \quad g, h_\alpha \in \mathcal{D}_n^{m-1} \quad  A \in \mathbb{N} \right\rbrace
\end{align*}
for each $m>0$. Let $\mathcal{D}_n := \bigcup_{m \in \mathbb{N}}\mathcal{D}_n^m$. We call \emph{weakly nonlocal} (\emph{WNL}) \emph{functional over $\mathbb{R}^n$} every functional of the form $I \circ f$ where $f \in \mathcal{D}_n$. We will write:
\begin{align*}
\tilde{\mathcal{W}}^m_n := I \circ \tilde{\mathcal{D}}^m_n \quad ; \quad \mathcal{W}^m_n := I \circ \mathcal{D}^m_n \quad ; \quad \mathcal{W}_n := I \circ \mathcal{D}_n
\end{align*}
\end{definizione}
\begin{remark}
In the definition above, the case $A=0$ is not excluded, hence we have $\mathcal{D}^i_n \subset \mathcal{D}^j_n$ whenever $i < j$.
\end{remark}
\begin{remark}
For a functional $F \in \tilde{\mathcal{W}}_n^m$ we can find an explicit representative for its density: it will be of the form
\begin{align}
\label{rep}
g \prod_{\alpha_1=1}^A d^{-1} 
\left( 
... 
\left(
h_{\alpha_1, ...,\alpha_m} \prod_{\alpha_{m+1}=1}^{A_{\alpha_1, ..., \alpha_{m}}} d^{-1} 
\left(
h_{\alpha_1, ..., \alpha_{m+1}} 
\right)
\right)
... 
\right)
\end{align}
where $g$ and all the $h$'s are local densities and some of the $A$'s can be zero.
\end{remark}

\begin{remark} \label{noproduct}
Notice that $\mathcal{D}_n$ has a natural structure of $\mathbb{R}$-algebra w.r.t. the obvious linear structure and the pointwise product $(f \cdot g )(u) := f(u) \cdot g(u)$.
On the other hand, it's important to remark that the space of weakly nonlocal functionals doesn't have the structure of an $\mathbb{R}$-algebra, meaning that there is no reasonable (in our context) way to define a product of WNL functionals. Let's analyze some examples:
\begin{itemize}
    \item Let's define our "product" of $F,G \in \mathcal{W}_n$ pointwise by setting $F \cdot G[u] := F[u] \cdot G[u]$. The reason why this definition doesn't work is that in general this is not an integral functional anymore, as there is no way to express it as the integral of a density.
    
    \item One may be tempted to define the product by taking the product of densities: given $F,G \in \mathcal{W}_n$ having $f, g$ as WNL densities , we could define $F \cdot G[u] := \int_\mathbb{R} f \cdot g(u) dx$. Unfortunately, this approach has its problems too:
    the space of WNL functionals can be realized as the quotient space of the space $\mathcal{D}_n$ of WNL densities by subspace $d \mathcal{D}_n$. This means that, in order to be defined over functionals, our product has to pass to equivalence at the level of densities. But in general
    \begin{equation*}
        (f + dh)\cdot(g + dk) = f \cdot g + f \cdot dk + dh \cdot g + dh \cdot dg \nsim f \cdot g
    \end{equation*}
\end{itemize}
\end{remark}
We want to extend the formula for the variational derivative to these new functionals. First of all, we consider the simplest nonlocal case: the one of $\mathcal{W}^1_n$. Let's remark that the following version of the Leibniz rule holds as a consequence of the Taylor formula.
\begin{lemma}
Consider $F \in \tilde{\mathcal{W}}^1_n$. It's is G-differentiable and $\forall k \in \mathcal{S}_n$
\begin{align*}
&d_G F (u)[k] = \int_\mathbb{R} \dfrac{\partial g}{\partial u^{(i)}}(u) \cdot k^{(i)} \prod_{\alpha = 1}^A (d^{-1} h_\alpha (u)) dx \\ 
&+ \sum_{\alpha =1}^A \int_\mathbb{R} g(u) d^{-1} \left( \dfrac{\partial h_{\alpha}}{\partial u^{(i)}}(u)\cdot k^{(i)}\right) \prod_{\beta \neq \alpha}(d^{-1} h_\beta(u)) dx
\end{align*}
\end{lemma}
\begin{proof}
Let's consider the case $n=A=1$; the general case is proven analogously. We have $F := I \circ (g \cdot d^{-1} h)$ and let $\phi \in C^\infty(\mathbb{R}^{2 + N})$, $\psi \in C^\infty(\mathbb{R}^{2 + M})$ be the N-local and M-local densities associated to $g$ and $h$ respectively.  Let's write, w.r.t. these two functions, the limit defining the G-differential of $F$ and develop the factors up to second order through the Taylor formula with Lagrangian remainder:
\begin{align*}
&\int_\mathbb{R} \dfrac{\partial g}{\partial u^{(i)}}(u) k^{(i)} \, d^{-1} h(u) dx + \int_\mathbb{R} g(u) \, d^{-1}\left( \dfrac{\partial h}{\partial u^{(i)}}(u) k^{(i)} \right)dx \\ 
&+ \lim_{t \rightarrow 0} \frac{t}{2} \int_\mathbb{R} \frac{\partial^2 \phi}{\partial u^{(i)} \partial u^{(j)}}(x,y_t(x)) k^{(i)}k^{(j)} d^{-1}h(u) dx \\
&+ \lim_{t \rightarrow 0} \frac{t}{2} \int_\mathbb{R} g(u) d^{-1} \left( \frac{\partial^2 \psi}{\partial u^{(i)} \partial u^{(j)}}(x,z_t(x)) k^{(i)}k^{(j)} \right) dx \, + \, ...
\end{align*}
The three dots hide four more terms that can be easily treated in the same way of the two explicitly written. To conclude the proof it's enough to show that the two integrals in the limits are bounded by a constant when $y_t$ and $z_t$ vary.
First of all notice that
\begin{align*}
s^{i j}(u) := k^{(i)}k^{(j)} d^{-1}h(u)
\end{align*}
is an $\mathcal{S}_1$ function.
The crucial fact is that $y_t(x)$ and $z_t(x)$, for each $t$ and $x$, always belong to the bounded set
\begin{align*}
B := \prod_{i=0}^N \left( Im( u^{(i)} ) + B_0(\Vert k \Vert_\infty ) \right)
\end{align*}
So thanks to the boundedness property of derivatives of local densities there are positive real numbers $M_{i j}(B)$ such that for each $t \in (-1,1)$
\begin{align*}
\left \vert \frac{\partial^2 \phi}{\partial u^{(i)} \partial u^{(j)}}(x, y_t) s^{i j}(u)(x)\right\vert \leq
M_{i j}(B) \left \vert s^{i j}(u)(x) \right\vert
\end{align*}
and the last function is in $\mathcal{S}_1 \subset L^1(\mathbb{R})$. This shows that the first limit is zero. The situation for the second one is very similar and can easily be recovered adapting the argument above.
\end{proof}

This lemma brings us to the following result, which gives, combined with the Euler-Lagrange formula, the general form of the G-differential of a $\tilde{\mathcal{W}}_n^1$ functional (by linearity this extends to every element of $\mathcal{W}^1_n$).
\begin{theorem}
Let $F \in \tilde{\mathcal{W}}_n^1$ and consider $v \in \mathcal{S}_n$. Then F is G-differentiable and defined $\frac{\delta F}{\delta v_l(x)} :=  R(x) + \sum_{\alpha=1}^A T_\alpha(x)$ 
with
\begin{align*}
R &:= \left( - \dfrac{d}{dx} \right)^i \left[ \dfrac{\partial g}{\partial {u_l}^{(i)}}(v) \cdot \prod_{\alpha = 1}^A d^{-1}(h_\alpha(v)) \right]\\
T_\alpha &:= - \left( - \dfrac{d}{dx} \right)^k \left[ d^{-1}\left(g(v) \cdot \prod_{\beta \neq \alpha} d^{-1}(h_\beta(v))\right) \dfrac{\partial h_\alpha}{\partial {u_l}^{(k)}}(v) \right]
\end{align*}
we have
\begin{align*}
d_G F(u)[k] = \int_\mathbb{R} \dfrac{\delta F}{\delta u(x)} k(x) dx
\end{align*}
for each $k \in \mathcal{S}_n$.
In this formula $M_\alpha$ is the order of the local density $h_\alpha$ and N the one of $g$.

\begin{proof}
As we did before, we work out the proof for the case $n=1$. This result follows as a consequence of the integration by parts of the integrals appearing in the statement of the previous lemma. From the first integral we quickly find $R$, so we consider the latter. Fixed $i$, integrating it by parts we get
\begin{align*}
&\int_\mathbb{R} g(u) d^{-1} \left( \dfrac{\partial h_{\alpha}}{\partial u^{(i)}}(u)\cdot k^{(i)}\right) \prod_{\beta \neq \alpha}(d^{-1} \circ h_\beta(u)) dx 
\\
&= \left. d^{-1} \left( \dfrac{\partial h_{\alpha}}{\partial u^{(i)}}(u)\cdot k^{(i)}\right)d^{-1}\left(g(u)\prod_{\beta \neq \alpha}(d^{-1} \circ h_\beta(u))\right) \right\vert_{-\infty}^{+\infty}\\
&- \int_\mathbb{R} d^{-1} \left( g(u) \prod_{\beta \neq \alpha}(d^{-1} \circ h_\beta(u)) \right)  \dfrac{\partial h_{\alpha}}{\partial u^{(i)}}(u)\cdot k^{(i)} dx 
\end{align*}
where the boundary term vanishes.
Integrating by parts $i$-times lowering the order of the derivative of $k$ we get $T_\alpha$.
\end{proof}
\end{theorem}

With a completely analogous proof using the representation (\ref{rep}) one finds the G-differentiability of general WNL functionals and obtains a formula for their variational derivative. In order to keep a readable notation without loosing any conceptual point, we write this formula only for functionals $F$ having density of the type
\begin{equation}
\label{nonlocalchains}
g \prod_{\alpha = 1}^A d^{-1} \left( h_{\alpha, 1}d^{-1} \left( ... \left( h_{\alpha, D_\alpha-1}d^{-1} \left( h_{ \alpha, D_\alpha} \right)\right)... \right)\right) 
\end{equation}
for such a functional we obtain $\frac{\delta F}{\delta v_l(x)} :=  R(x) + \sum_{\alpha = 1}^A \sum_{\delta = 1}^{D_\alpha} T^\delta_\alpha(x)$ where (omitting all evaluations in $v$)
\begin{align*}
&R := \left( - \dfrac{d}{dx} \right)^i \left[ \dfrac{\partial g}{\partial {u_l}^{(i)}} \cdot \prod_{\alpha=1}^A H_\alpha \right]\\
&T^\delta_\alpha :=(-1)^\delta \left( - \dfrac{d}{dx} \right)^k \left[ \check{H}^\delta_\alpha \left( g \cdot \prod_{\beta \neq \alpha} H_\beta \right) \, \dfrac{\partial h_{\alpha, \delta}}{\partial {u_l}^{(k)}} \, \widehat{H}^{\delta+1}_\alpha \right]
\end{align*}
where we have defined
\begin{flalign*}
&\widehat{H}^\delta_\alpha := d^{-1} \left( h_{\alpha, \delta}d^{-1} \left( ... \left( h_{\alpha, D_\alpha-1}d^{-1} \left( h_{D_\alpha} \right)\right)... \right)\right) && \\
&\check{H}^\delta_\alpha(*) := d^{-1}\left( h_{\alpha, \delta-1} d^{-1}\left(...  \left(h_{\alpha, 1} d^{-1} \left( * \right)\right)...\right)\right)&&\\
& H_\alpha := \widehat{H}_\alpha^1
\end{flalign*}
Here $N$ is the order of $g$ and $M_\alpha^\delta$ the one of $h_{\alpha, \delta}$. Thanks to these computations we get the following
\begin{corollario}
Let $F \in \mathcal{W}_n$. Then its variational derivative w.r.t. every $v \in \mathcal{S}_n$ is bounded.
\end{corollario}
\begin{proof}
For simplicity, we will work out the proof for functionals having densities of the form (\ref{nonlocalchains}) and for $n=1$. In the whole computation we omit the evaluation at $v$ of all the local densities.
Consider first the part given by $R$. We claim that $\left( \dfrac{d}{dx} \right)^i \left[ \dfrac{\partial g}{\partial {u_l}^{(i)}} \cdot \prod_{\alpha=1}^A H_\alpha \right]$ is bounded. We have by Leibniz rule that the expression above is equal to
\begin{align*}
\sum_{j=0}^i \binom{j}{i} \left(\frac{\partial g}{\partial {u_l}^{(i)}}\right)^{(j)}\left(\prod_{\alpha=1}^A H_\alpha\right)^{(i-j)}
\end{align*}
Now $({\partial g}/{\partial {u_l}^{(i)}})^{(j)}$, by chain rule, can be written as finite sum of partial derivatives of $g$, some of which are multiplied by a derivative of $v$. We have by definition that those partial derivative are bounded and $v$ is a Schwartz function, so the whole sum is bounded. On the other hand the term $\left(\prod_{\alpha=1}^A H_\alpha\right)^{(i-j)}$ is bounded too, as $\prod_{\alpha=1}^A H_\alpha$ is a product of bounded functions with bounded derivatives of all orders. This holds because
\begin{align*}
\dfrac{d}{dx}H_\alpha = h_{\alpha, 1}d^{-1} \left( ... \left( h_{\alpha, D_\alpha-1}d^{-1} \left( h_{ \alpha, D_\alpha} \right)\right)... \right) \in \mathcal{S}_n
\end{align*}
by the third property of $d^{-1}$ highlighted after its definition. For what concerns $T_\alpha^\delta$, the argument for proving that it's bounded is essentially the same we used above for $R$.
\end{proof}

\section{FROM LOCAL TO GLOBAL:\\ FUNCTIONALS ON MANIFOLDS AND POISSON BRACKETS}\label{Sect3}

Let's now give a global interpretation to these functionals. The main reason that motivates this kind of globalization process is that some properties of the brackets we are going to study, such as their relationship with the Riemannian geometry of Euclidean hypersurfaces, are easily understood once we think of the brackets as global objects. 
I order to perform the needed constructions we have to shift our attention toward geometry, focusing on manifolds modeled on infinite dimensional spaces. The theory of such manifolds is very rich and well studied: see for example (Ref. \cite{Lang}) for a discussion of Banach manifolds and (Ref. \cite{Hamilton}) for the case of Fréchet manifolds. We will only need some elementary constructions, so we describe them explicitly in this section.
\begin{definizione}
Let $\Omega \subseteq \mathbb{R}^n$ an open neighborhood of the origin. We define the open set $\mathcal{S}(\Omega) \subseteq \mathcal{S}_n$ of the Schwartz functions having image in $\Omega$. We call \emph{local densities on $\Omega$} the restrictions of local densities to $\mathcal{S}(\Omega)$ and \emph{local functionals on $\Omega$} the compositions of the integral functional $I$ with a local density on $\Omega$. Analogously, we define the \emph{WNL functionals on $\Omega$} as before where all the local densities appearing in the previous definitions are now local densities on $\Omega$. The spaces of these functionals will be denoted $\mathcal{L}(\Omega)$ and $\mathcal{W}(\Omega)$.
\end{definizione}
\begin{remark}
	Notice that $\mathcal{S}(\Omega)$ is an open subset of the Fréchet space $\mathcal{S}_n$.
\end{remark}
\begin{remark}
	W.r.t our previous notation, we have $\mathcal{S}(\mathbb{R}^n) = \mathcal{S}_n$,\, $\mathcal{L}(\mathbb{R}^n) = \mathcal{L}_n$ and $\mathcal{W}(\mathbb{R}^n) = \mathcal{W}_n$.
\end{remark}

Let $M$ be a smooth connected finite dimensional manifold and fix $y \in M$. We want this $y$ to play the role of the origin in our manifold, following the approach outlined in (Ref.~\cite{MaltNov}).
\begin{definizione}
	Let $\mathcal{A} := \left\lbrace (U_\lambda, \varphi_\lambda)\right\rbrace_{\lambda \in \Lambda}$ be the subset of the maximal atlas of $M$ such that $U_\lambda$ is connected,  $y \in U_\lambda$ and $\varphi_\lambda(y) = 0$ for each $\lambda \in \Lambda$. We will say that $\mathcal{A}$ is the \textit{maximal atlas for the pointed manifold $(M,y)$}.
\end{definizione} 
We point out that the submanifold of $M$ covered by $\mathcal{A}$ is the the whole $M$ itself. This is proven in proposition \ref{HomManifolds}. We split this result in three parts.
\begin{lemma}
	Let $\mathbb{B}\subset \mathbb{R}^n$ be a convex open subset. Then for every $a, b \in \mathbb{B}$ there is $\phi \in \text{Aut}(\mathbb{R}^n)$ such that $\phi(a) = b$ and $\phi_{\vert \mathbb{B}^c} = \mathbb{1}_{\mathbb{B}^c}$.
\end{lemma}
\begin{proof}
	Consider an open, relatively compact subset $W$ of $\mathbb{B}$ containing the segment joining $a$ and $b$.
	Let $f \in C_c^\infty(\mathbb{R}^n)$ be a bump function supported in $\mathbb{B}$ such that $f_{\vert W} = 1$. Consider the compactly supported smooth vector field $V$ on $M$ defined by $V(x):= f(x)(b-a):= f(x)\sum_{i=1}^n (b_i - a_i)\frac{\partial}{\partial x^i}$. This is a complete vector field, so it admits a one parameter group of automorphisms $\left\lbrace \varphi_t \right\rbrace_{t \in \mathbb{R}}$. We can consider the automorphism $\phi:= \varphi_1$. It sends $a$ to $b$ and it's clearly the identity outside $\mathbb{B}$.
\end{proof}
To simplify the notation a bit, given a smooth manifold $M$ and an open subspace $T$, we define $\text{Aut}_{T}(M) := \left \lbrace \phi \in \text{Aut}(M) \quad s.t. \quad \phi_{\vert T^c} = \mathbb{1}_{T^c} \right \rbrace$. Notice that in our notation elements of $Aut_T(M)$ fix the complement of T and not T itself. Moreover, given a chart $(U,h)$ for $M$, we will say that it's \textit{convex} if $h(U)$ is a convex subset of $\mathbb{R}^n$ and that it's a \textit{$T$-chart} if $U\subseteq T$.
Notice that for each $x \in T$ there exists a nonempty convex $T$-chart around it.
\begin{proposizione}
	Let $M$ be an $n$-dimensional smooth manifold. For every two points $x,y \in M$ and for each connected open $T$ containing them there is an automorphism $\phi \in \text{Aut}_T(M)$ such that $\phi(x)=y$.
\end{proposizione}
\begin{proof}
	At first, we need some additional assumptions: We assume that there exists a convex $T$-chart $(U,h)$ such that $x,y \in U$. We can pick a smaller open convex $\mathbb{B}\subset h(U)$ containing the images of $x$ and $y$. Let $\phi \in Aut(\mathbb{R}^n)$ be the map defined in the lemma above for the choice $a:= h(x)$ and $b:= h(y)$. This is the identity outside $\mathbb{B}$, hence $h^{-1} \circ \phi \circ h \in \text{Aut}(U)$ can be extended to the whole $M$ by letting it be the identity on $U^c$. This is an automorphism of $M$ satisfying the required conditions. Now, in order to prove the statement in the general case, consider an open neighborhood $T$ of $y$. We define
	\begin{align*}
		W_{T,y} := \left\lbrace x \in T : \quad \exists \phi \in \text{Aut}_T(M) \quad s.t. \quad 
		\phi(x)=y \right\rbrace \subseteq T
	\end{align*}
	For a moment we forget about $M$ and we consider the ambient manifold to be $T$. Then
	\begin{itemize}
	\item $W_{T,y}$ is open: let $x \in W_{T,y}$, consider a convex $T$-chart $(U,h)$ around it and pick $z \in U$. Clearly $\exists \psi \in \text{Aut}_U(M)$ mapping $z$ to $x$ by the first part of the proof. We have by assumption an automorphism $\phi \in \text{Aut}_T(M)$ sending $x$ to $y$. Then $\phi \circ \psi$ is in $\text{Aut}_T(M)$ and $\phi \circ \psi(z) = y$, hence $z \in W_{T,y}$.
	\item $W^c_{T,y}$ is open: let $x \in W^c_{T,y}$, consider a convex $T$-chart $(U,h)$ around it and pick $z \in U$. Assume by contradiction that there is an automorphism $\psi \in \text{Aut}_T(M)$ sending $z$ to $y$. By the first part of the proof $\exists \eta \in \text{Aut}_U(M)$ mapping $x$ to $z$. Then $\psi \circ \eta\in \text{Aut}(M)$ is such that $\phi \circ \psi(x) = y$, hence $x \in W_{T,y}$. This is a contradiction, so $U \subseteq W^c_{T,y}$.
	\end{itemize}
	So $W_{T,y}$ is nonempty (containing $y$) and both open and closed. $T$ is connected, so $W_{T,y}=T$. The fact that this holds for every $T$ and $y$ is exactly our claim.
\end{proof}
\begin{proposizione}\label{HomManifolds}
	Let $M$ be a smooth connected $n$-dimensional manifold. Then for every $x, y \in M$ there is a connected coordinate patch $(U,h)$ such that $x, y \in U$.
\end{proposizione}
\begin{proof}
	If $n=1$ we have by (Ref. \cite{Milnor}) that $M$ is diffeormorphic either to an open interval of the real line or to $S^1$. In both cases the result is trivial. Assume now that $n$ is bigger than 1.
	Let $(U,h)$ be a connected chart around $x$ and consider $z \in U\setminus\lbrace x \rbrace$. By $n \geq 2$ we know that $M\setminus\lbrace x \rbrace$ is connected, so we have an automorphism $\phi \in \text{Aut}_{M\setminus\lbrace x \rbrace}(M)$ sending $y$ to $z$ by the proposition above. This precisely means that $\phi(x) = x$ and $\phi(y) = z$. Now consider the diffeomorphism $h \circ \phi : \phi^{-1}(U) \rightarrow \mathbb{R}^n$.This gives us a chart $(\phi^{-1}(U), h \circ \phi)$ containing both $x$ and $y$. 
\end{proof}
The kind of Fréchet manifolds we are interested in are sometimes called loop spaces of smooth manifolds. The standard reference for their theory is the book (Ref. \cite{KrieglMich}) of A. Kriegl and P. W. Michor. A more concise treatment of the subject can be found in the paper (Ref. \cite{Stacey}) of A. Stacey.
The idea of the following construction is to build a Fréchet manifold $M^\star_y$, modeled on $\mathcal{S}_n$, whose points are Schwartz functions valued in some coordinate patch $(U, \varphi)$ of $M$ such that $y \in U$ and $\varphi(y)=0$. We will call this new infinite dimensional manifold $M^\star_y$ the \textit{loop space over} $(M,y)$. We won't care about the smooth structure of this space (and indeed we won't even use its topology) so we limit ourselves to defining it at the topological level, as a $C^0$-Fréchet manifold.
\begin{remark}
	The name "loop space" comes from the fact that the kind of functions we consider can be regarded as functions from $S^1$ to $M$. Indeed, these functions are such that their values and those of all of their derivatives go to zero when the variable goes to $\pm \infty$. Intuitively, by adding to the domain the point at infinity and extending the map by sending such point to zero we obtain a well defined map from the projective line (and hence $S^1$) to $M$. 
\end{remark}
The idea of using the spaces of loops $S^1 \rightarrow M$ to provide a setting for studying Poisson brackets appears, for example, in the paper (Ref. \cite{Mokhov}) of O. Mokhov. Let's now construct the loop space.
First of all, we will simplify the notation a bit; it's easy to lose track of all the objects we introduced, so we will recall some definitions below.
\begin{itemize}
	\item $\mathcal{A} := \left\lbrace (U_\lambda, \varphi_\lambda)\right\rbrace_{\lambda \in \Lambda}$ will be the maximal atlas for the pointed manifold $(M,y)$.
	\item We will write $\tilde{U}_\lambda$ instead of $\varphi_\lambda(U_\lambda) \subseteq \mathbb{R^n}$.
	\item The symbol $\mathcal{S}^\lambda$ will be used instead of $\mathcal{S}(\tilde{U}_\lambda)$. Recall that by definition
	\begin{align*}
		\mathcal{S}^\lambda := \left \lbrace f : \mathbb{R} \rightarrow \tilde{U}_\lambda \quad : \quad f_i \in \mathcal{S}_1 \quad \forall i \in \lbrace 1,  ..., n \rbrace  \right \rbrace \subseteq \mathcal{S}_n
	\end{align*}
	\item In a completely analogous way we will write $\mathcal{L}^\lambda := \mathcal{L}(\tilde{U}_\lambda)$ and $\mathcal{W}^\lambda := \mathcal{W}(\tilde{U}_\lambda)$
\end{itemize}  
In order to build the topological Fréchet manifold $M_y^\star$ we adopt the strategy suggested by the following classical, well known result. What this result tells us is that it's always possible to reconstruct a manifold from a given atlas (the analogue for finite dimensional manifolds is proven, for example, in Lemma 1.35 of the book (Ref. \cite{Lee}) of J. M. Lee). In order to make the construction a bit clearer we will write a proof of the result.
\begin{proposizione}
	Assume we have a family $\left\lbrace \mathcal{O}_i \right \rbrace_{i \in I}$ of open subsets of Fréchet spaces such that, for each ordered couple $\mathcal{O}_i, \mathcal{O}_j$ of these sets, there is an open $\mathcal{O}_{i j} \subseteq \mathcal{O}_i$ and a homeomorphism of Fréchet spaces $\varphi_{j i}: \mathcal{O}_{i j} \rightarrow \mathcal{O}_{j i}$. Assume moreover that these homeomorphisms satisfy the usual cocycle conditions, namely $\varphi_{i i} = \mathbb{1}_{\mathcal{O}_i}$ and $\varphi_{k i} = \varphi_{k j} \circ \varphi_{j i}$ $\forall i, j, k \in I$.
	If $\forall x \in \mathcal{O}_i, \forall y \in \mathcal{O}_j$ one of
	\begin{align}\label{Hausdorff}
		&\exists k \in I \quad : \quad x \in \mathcal{O}_{i k} \text{  and  } y \in \mathcal{O}_{j k} \nonumber\\ 
		&\exists k_1, k_2 \in I \quad : \quad x \in \mathcal{O}_{i k_1}, y \in \mathcal{O}_{j k_2}\, \text{  and  } \,\varphi_{k_1 i}{x} \notin \partial \mathcal{O}_{k_1 k_2} \subset \mathcal{O}_{k_1}
	\end{align}
	holds, then there exists a (unique up to homeomorphism) Fréchet manifold $B$ together with an atlas $\left \lbrace (W_i, \phi_i)\right \rbrace_{i \in I}$ such that
	\begin{itemize}
		\item $\phi_i : W_i \rightarrow \mathcal{O}_i$ is an homeomorphism and $\phi_i(W_i \cap W_j) = \mathcal{O}_{i j}$.
		\item For each couple of charts $W_i$, $W_j$ the transition function is exactly $\varphi_{j i} : \mathcal{O}_{i j} \rightarrow \mathcal{O}_{j i}$.
	\end{itemize}
\end{proposizione}
\begin{proof}
For what concerns existence, define the topological space
\begin{align*}
B := \bigslant{\bigsqcup_{i \in I} \mathcal{O}_i}{\sim}
\end{align*}
where $\sim$ is the relation on the disjoint union defined by $(t)_i \sim (s)_j \iff t \in \mathcal{O}_{i j} \,\, \land \,\, \varphi_{j i}(t) = s$. This is an equivalence relation by the cocycle conditions and this space is Hausdorff because of (\ref{Hausdorff}).
We can consider the subsets $B_i := \pi_i(\mathcal{O}_i)$ that cover $B$, where $\pi_i:\mathcal{O}_i \rightarrow B$ is composition of the inclusion in the disjoint union and the projection onto $B$. Clearly $\pi_i$ is an injection. Then the atlas of $B$ is given by the functions $\pi_i^{-1}: B_i \rightarrow \mathcal{O}_i$ inverted after restricting the target of $\pi_i$ to $B_i$.
We have that $\pi_i^{-1}( B_i \cap B_j ) = \left \lbrace x \in \mathcal{O}_i \quad : \quad \exists y \in \mathcal{O}_j \quad s.t \quad (x)_i \sim (y)_j \right \rbrace = \mathcal{O}_{i j}$.
The definition of our equivalence relation gives us the result regarding the transition functions, in fact
\begin{align*}
	\pi_j^{-1} \circ \pi_i (x) = \pi_j^{-1}[(x)_i] = \pi_j^{-1}[(\varphi_{ji}(x))_j] = \varphi_{ji}(x)
\end{align*}
To prove uniqueness consider another such Fréchet manifold $C$ with an atlas $\lbrace(V_i, \psi_i)\rbrace_{i \in I}$ satisfying the conditions above. Then the various homeomorphisms $\phi_i^{-1} \circ \psi_i: V_i \rightarrow W_i$  agree on the overlaps and can therefore be glued to a global homeomorphism between the two Fréchet manifolds $C$ and $B$. 
\end{proof}
We will apply the previous result to the following data:
\begin{enumerate}
\item We have the family $\left\lbrace \mathcal{S}^\lambda \right\rbrace_{\lambda \in \Lambda}$ of open subsets of the Fréchet space $\mathcal{S}_n$.
\item For each pair $\lambda, \mu \in \Lambda$ we have an open subset
\begin{align*}
	\mathcal{S^{\lambda \mu}}:= \left\lbrace u \in \mathcal{S}^\lambda : \quad \varphi_\lambda^{-1} \circ u (\mathbb{R}) \subset U_\mu \right\rbrace \subseteq \mathcal{S}^\lambda
\end{align*} 
\item Moreover for every couple $\lambda, \mu \in \Lambda$ there is a function 
\begin{align*}
\varphi_{\mu \lambda}: \mathcal{S}^{\lambda \mu} \rightarrow \mathcal{S}^{\mu \lambda} \quad \quad s.t. \quad \quad u \longrightarrow \varphi_\mu \circ {\varphi_\lambda}^{-1} \circ u
\end{align*}
\end{enumerate}
In the next lines, and in particular in the following two propositions, we will show that the operators $\varphi_{\mu \lambda}$ are well defined and continuous. 
\begin{remark}
	Notice that, considering the different components separately, it's enough to focus on operators of the form $\mathcal{S}_n \ni f \rightarrow \phi \circ f \in \mathcal{S}_1$ for $\phi \in C^\infty(\mathbb{R}^n, \mathbb{R})$.
\end{remark}
We will perform some computations in Schwartz spaces. A standard reference for the topics considered in this discussion is the book (Ref. \cite{ReedSimon}) of M. Reed and B. Simon.
Explicitly, the norms and metric on $\mathcal{S}_1$ are of the form
\begin{align*}
\Vert f \Vert_{\alpha, k} := \left\Vert x^\alpha \cdot \frac{d^k f}{d x^k} \right\Vert_\infty \qquad \qquad d(f,g) := \sum_{\alpha, k \in \mathbb{N}} 2^{-\alpha -k}\frac{\Vert f - g \Vert_{\alpha, k}}{1 + \Vert f - g \Vert_{\alpha, k}}
\end{align*}
for $\alpha, k \in \mathbb{N}$.
In the spaces $\mathcal{S}_n$ we consider the product norms and metric.
From the explicit expressions above it's immediate to check that the metric in these spaces is translation-invariant.
Assume we have a smooth map $\phi\in C^\infty(\mathbb{R^n})$ such that $\phi(\underline{0})=0$ and consider the operator $C_\phi : \mathcal{S}_n\rightarrow\mathcal{S}_1$ defined by $C_\phi(f) := \phi \circ f$.
\begin{proposizione}
	$C_\phi$ is well defined as its range is in $\mathcal{S}_1$.
\end{proposizione}
\begin{proof}
	We have to prove $\Vert \phi \circ f \Vert_{\alpha, k} < \infty$ for each $\alpha, k \in \mathbb{N}$.
	We start by considering the case without derivatives: assume $k=0$. Consider the following limit, computed using the De L'H\^{o}pital rule
	\begin{align*}
	\lim\limits_{x \to \pm\infty} x^\alpha \phi(f(x)) = \lim\limits_{x \to \pm\infty} \frac{\phi(f(x))}{x^{-\alpha}} = - \frac{1}{\alpha}\sum_{i=1}^{n} \lim\limits_{x \to \pm\infty}  \frac{\partial \phi}{\partial x_i}(f(x)) \frac{f_i^\prime(x)}{x^{-\alpha-1}} = 0
	\end{align*}
	being $f^\prime_i \in \mathcal{S}_1$ and the partial derivatives of $\phi$ bounded in an neighborhood of $\underline{0}$. This allows us to show that $ x^\alpha \phi \circ f$ tends to zero as $x$ goes to infinity. Being this map continuous, it has to be bounded on all the real line, so $\Vert \phi \circ f \Vert_{\alpha, 0} < \infty$ for every $\alpha \in \mathbb{N}$.
	Now assume $k>1$. By the multivariate Faà di Bruno formula (Ref. \cite{ConstSav}, corollary 11) we have the following formula to compute the high order derivatives of the composition:
	\begin{align*}
	\frac{d^k (\phi \circ f)}{d x^k}(x) = \sum_{1 \leq \vert \lambda \vert \leq k}  \frac{\partial^\lambda \phi}{\partial x^\lambda}(f(x)) \cdot P_\lambda[f](x)
	\end{align*}
	where $\lambda \in \mathbb{N}^n$ and $P_\lambda[f]$ is a polynomial function (without constant terms) in the derivatives up to order $\vert \lambda \vert$ of the various components of $f$. We have adopted the notation $\frac{\partial^\lambda}{\partial x^\lambda} := \frac{\partial^{\vert\lambda\vert}}{\partial x_1^{\lambda_1} \,\cdot \cdot \cdot\, \partial x_n^{\lambda_n}}$. Let's highlight one important detail that follows from boundedness of Schwartz functions: $\overline{\text{Im }f}$ is compact in $\mathbb{R}^n$.
	Then we see that by smoothness of $\phi$ we can define $V_\lambda := \sup_{z \in Im(f)} \left \vert \frac{\partial^\lambda \phi}{\partial x^\lambda}(z)\right\vert < \infty$ and then
	\begin{align*}
	\left\Vert x^\alpha \cdot \frac{d^k (\phi \circ f)}{d x^k}(x) \right\Vert_\infty \leq \sum_{1 \leq \vert \lambda \vert \leq k}  V_\lambda \left \Vert x^\alpha P_\lambda[f] \right\Vert_\infty \leq \sum_{1 \leq \vert \lambda \vert \leq k}  V_\lambda \left \Vert P_\lambda[f] \right\Vert_{\alpha, 0} < \infty
	\end{align*}
	The last inequality holds by the case $k=0$ described at the beginning of the proof, choosing $\phi:= P_\lambda$. The fact that this quantity is finite for each $\alpha, k \in \mathbb{N}$ shows that $\phi \circ f \in \mathcal{S}_1$.
\end{proof}
At this point we can focus on continuity.
\begin{proposizione}
	$C_\phi$ is a continuous operator.
\end{proposizione}
\begin{proof}
	First of all recall that in order to prove continuity w.r.t. the Fréchet metrics it's enough to prove it once we fix on the target the topologies induced by the norms $\Vert \cdot \Vert_{\alpha, k}$. We start, as in the previous proposition, by considering the case of the norms $\Vert \cdot \Vert_{\alpha, 0}$.
	Consider $f \in \mathcal{S}_n$. By smoothness of $\phi$ and boundedness of $\overline{\text{Im }f}$ we have that $\exists \delta, K >0$ such that $\forall z, w \in \text{Im }f + B(\underline{0}, \delta)$ then $\vert \phi(w) - \phi(z)\vert < K \sum_{i=1}^n \vert w_i - z_i \vert$, by Lipschitzianity of $\phi$ on compact sets. Then we choose $\sigma \in \mathcal{S}_n$ such that $d(0,\sigma) < \delta $. We have that for every $\alpha \in \mathbb{N}$
	\begin{align}
	\left\Vert x^\alpha \left( \phi \circ (f + \sigma) - \phi \circ f \right) \right\Vert_\infty < 
	K \left\Vert x^\alpha \sum_{i=1}^n \vert \sigma_i \vert \right\Vert_\infty \leq
	K \sum_{i=1}^n \Vert\sigma_i \Vert_{\alpha, 0} < K 2^\alpha \frac{\delta}{1-\delta}
	\end{align}
	We can now consider the cases with $k>0$.
	For each $\epsilon > 0$ and $f, \sigma \in \mathcal{S}_n$, using of the Faà di Bruno formula written above we obtain the estimate 
	\begin{align*}
	&\left\Vert x^\alpha \left(\frac{d^k (\phi \circ (f+\sigma))}{d x^k} - \frac{d^k (\phi \circ f)}{d x^k} \right) \right\Vert_\infty\\ \leq
	& \sum_{1 \leq \vert \lambda \vert \leq k} \left\Vert x^\alpha \left(\frac{\partial^\lambda \phi}{\partial x^\lambda} \circ (f + \sigma) - \frac{\partial^\lambda \phi}{\partial x^\lambda}\circ f \right) P_\lambda[f+\sigma] \right\Vert_\infty +
	\sum_{1 \leq \vert \lambda \vert \leq k} \left\Vert x^\alpha \left(\frac{\partial^\lambda \phi}{\partial x^\lambda}\circ f \right) \left(P_\lambda[f+\sigma] - P_\lambda[f] \right) \right\Vert_\infty
	\end{align*}
	By the usual argument involving relative compactness of $\text{Im }f$ and smoothness of $\phi$ we can see that there is a $\delta>0$ such that whenever $d(0,\sigma) < \delta$ then $\left\Vert \frac{\partial^\lambda \phi}{\partial x^\lambda}\circ(f+\sigma) - \frac{\partial^\lambda \phi}{\partial x^\lambda}\circ f \right\Vert_\infty \leq \epsilon$.
	Moreover the lemma above ensures that, up to picking a smaller $\delta$, we can assume
	$\Vert P_\lambda[f + \sigma] - P_\lambda[f] \Vert_{\alpha,0} < \epsilon$ for each $\sigma$ such that $d(f,\sigma)<\delta$.
	Define 
	\begin{align*}
	V_\lambda := \left\Vert \frac{\partial^\lambda \phi}{\partial x^\lambda}\circ f \right\Vert_\infty + \epsilon \qquad \qquad W_\lambda := \Vert P_\lambda[f] \Vert_{\alpha,0} + \epsilon
	\end{align*}
	Then the estimate above is smaller or equal than
	\begin{align*}
	\sum_{1 \leq \vert \lambda \vert \leq k} \epsilon \left\Vert P_\lambda[f+\sigma] \right\Vert_{\alpha, 0} +
	\sum_{1 \leq \vert \lambda \vert \leq k} V_\lambda \left\Vert P_\lambda[f+\sigma] - P_\lambda[f] \right\Vert_{\alpha, 0} 
	\leq \epsilon \sum_{1 \leq \vert \lambda \vert \leq k}  (W_\lambda + V\lambda) 
	\end{align*}
	which proves continuity.
\end{proof}
\begin{remark}
	Clearly the operator $\varphi_{ \mu \lambda}$ is bijective with inverse $\varphi_{\lambda \mu}$, so we have a family of homeomorphisms. Moreover, these maps satisfy the cocycle conditions: for each triplet $\lambda, \mu, \kappa \in \Lambda$  we have $\varphi_{\lambda \lambda} = \mathbb{1}_{\mathcal{S}^\lambda}$ and the following diagram commutes
	\begin{center}
		\begin{tikzpicture}
		\matrix (m) [matrix of math nodes, nodes in empty cells,row sep=3em,column sep=4em,minimum width=2em]
		{
			& \mathcal{S}^{\lambda \mu} \cap \mathcal{S}^{\lambda \kappa} &  \\
			\mathcal{S}^{\mu \lambda} \cap \mathcal{S}^{\mu \kappa} &  & \mathcal{S}^{\kappa \mu} \cap \mathcal{S}^{\kappa \lambda} \\};
		\path[-stealth]
		(m-1-2) edge node [above] {$\varphi_{ \mu \lambda}$} (m-2-1)
		(m-1-2) edge node [above] {$\varphi_{ \kappa \lambda}$} (m-2-3)
		(m-2-1) edge node [below] {$\varphi_{\kappa \mu}$} (m-2-3);
		\end{tikzpicture}
	\end{center}
To see this, just write $\varphi_{\kappa \mu} \circ \varphi_{\mu \lambda} = \varphi_k \circ \varphi_\mu^{-1} \circ \varphi_\mu \circ \varphi_\lambda^{-1} = \varphi_\kappa \circ \varphi_\lambda^{-1} = \varphi_{\kappa \lambda}$. 
\end{remark}
\begin{lemma}
	For every $\lambda, \mu \in \Lambda$ and $u \in \mathcal{S}^\lambda, v \in \mathcal{S}^\mu$, at least one of the two conditions (\ref{Hausdorff}) holds true.
\end{lemma}
\begin{proof}
	Define $\tilde{u}:=\varphi_\lambda^{-1}\circ x$,  $\tilde{v}:=\varphi_\mu^{-1}\circ y$. The two conditions (\ref{Hausdorff}) are rephrased in our context as:
	\begin{align*}
	&\exists \kappa \in \Lambda \quad : \quad Im(\tilde{u}),Im(\tilde{v}) \subset U_\kappa\\ 
	&\exists \kappa_1, \kappa_2 \in I \quad : \quad Im(\tilde{u}) \subset U_{\kappa_1}, Im(\tilde{v}) \subset U_{\kappa_2} \text{  and  } \, \varphi_{k_1} \circ \tilde{u} \notin \partial \mathcal{S}^{\kappa_1 \kappa_2} \subset \mathcal{S}^{\kappa_1}
	\end{align*}
	Assume that the first one doesn't hold. We can clearly assume $\tilde{U}_\lambda$ and $\tilde{U}_\mu$ are bounded in $\mathbb{R}^n$. Chosen $(\kappa_1, \kappa_2) := (\lambda, \mu)$, being $u$ in the boundary $\partial \mathcal{S}^{\lambda \mu}$ we have that for every $\epsilon>0$ there is $z \in \mathcal{S}^{\lambda \mu}$ such that $d(x, z) < \epsilon$. In particular, this has to hold if we replace the metric $d$ with the uniform metric. Then we obtain that $Im(x)$ is contained in $\overline{\tilde{U}_\lambda \cap \varphi_\lambda \circ \varphi_\mu^{-1}(\tilde{U}_\mu)}$ . So we have that $Im(\tilde{x}) \subset \overline{U_\mu}$ and obviously $Im(\tilde{x}) \cap \partial U_\mu \neq \emptyset$, otherwise the first condition would hold for $\kappa := \mu$. By compactness of $Im(\tilde{v})$ we can shrink $U_\mu$ obtaining an open neighborhood $W$ of $Im(\tilde{v})$ that is relatively compact in $U_\mu$. Then $(W, \varphi_\mu)$ is an element of $\mathcal{A}$, let's say the one corresponding to the index $\nu \in \Lambda$ and $(\kappa_1, \kappa_2) := (\lambda, \nu)$ satisfy the second condition.
\end{proof}
This allows us to build a gluing of this data, obtaining a Fréchet manifold by the proposition above. Explicitly, we have the manifold:
\begin{align*}
M^{\star}_y = \bigslant{\bigsqcup_{\lambda \in \Lambda} \mathcal{S}^\lambda}{\sim}
\end{align*}
where $(u)_\lambda \sim (w)_\mu$ if and only  if $u \in \mathcal{S}^{\lambda \mu}$, $w \in \mathcal{S}^{\mu \lambda}$ and $w = \varphi_\mu \circ \varphi_\lambda^{-1}(u)$.
The atlas has as open sets the images of the maps $\pi^\mu: \mathcal{S}^\mu \hookrightarrow \bigsqcup_{\lambda \in \Lambda} \mathcal{S}^\lambda \rightarrow \bigslant{\bigsqcup_{\lambda \in \Lambda} \mathcal{S}^\lambda}{\sim}$ and as maps the functions $\phi_\mu[(f)_\mu] = f$.

Intuitively, our definition of "global functional" over $M$ will be the one of a function on the disjoint union defined above that passes to the quotient by the relation $\sim$.
\begin{definizione}
A map $\tilde{F}: M^\star_y \rightarrow \mathbb{R}$ is called \emph{local functional on $(M,y)$} if, on the charts of the atlas for the loop space defined above, it is represented by a family of local functionals $\lbrace F_\lambda \in \mathcal{L}^\lambda\rbrace_{\lambda \in \Lambda}$.
\emph{WNL functionals on M} are defined in a completely analogous way. The spaces of these functionals will be denoted by $\mathcal{L}(M,y)$ and $\mathcal{W}(M,y)$ respectively.
\end{definizione}
In what follows we will identify these functionals with the families parameterized by $\Lambda$ that define them.
\begin{remark}
Notice that $\mathcal{L}(M,y)$ and $\mathcal{W}(M,y)$ have a natural structure of $\mathbb{R}$-linear spaces. The operations are defined at the level of the families that define the functionals.
\end{remark}
\begin{definizione}
We will call \emph{WNL Poisson bracket over $(M,y)$} a map
\begin{align*}
\lbrace \cdot , \cdot \rbrace : \mathcal{W}(M,y) \times \mathcal{W}(M,y) \rightarrow \mathcal{W}(M,y)
\end{align*}
which is bilinear and satisfies the following two identities:
\begin{align*}
&\lbrace F , G \rbrace = - \lbrace G , F \rbrace\\
&\lbrace \lbrace F , G \rbrace , H \rbrace + \lbrace \lbrace G , H \rbrace , F \rbrace + \lbrace \lbrace H , F \rbrace , G \rbrace = 0
\end{align*}
for each $F$, $G$, $H \in \mathcal{W}(M,y)$.
Moreover, we require it to have the form
\begin{equation}
\label{poisson}
\lbrace F , G \rbrace_\lambda (u) := \int_\mathbb{R}{\dfrac{\delta F_\lambda}{\delta u_i (x)} \left( P^{i j}_\lambda(u) \dfrac{\delta G_\lambda}{\delta u_j}\right) (x) dx}
\end{equation}
where
\begin{itemize}
\item $P^{i j}_\lambda(u)$, given $u \in \mathcal{S}^\lambda$, is the operator $\,$ $C_b^\infty(\mathbb{R}, \mathbb{R}^n) \rightarrow C_b^\infty(\mathbb{R}, \mathbb{R}^n)$ defined by 
\begin{align*}
P^{i j}_\lambda(u) &:= {g_\lambda(u)}^{i j} \dfrac{d}{dx} - {g_\lambda (u)}^{i s} {\Gamma_\lambda (u)}^j_{s k} u_x^k + {w_\lambda(u)}_k^i u_x^k d^{-1} {w_\lambda(u)}_l^j u_x^l
\end{align*} 
\item $g_\lambda, w_\lambda \in C^\infty(\tilde{U}_\lambda, \mathbb{R}^{n \times n})$, $\Gamma_\lambda \in C^\infty(\tilde{U}_\lambda, \mathbb{R}^{n \times n \times n})$ are such that the matrix $g_\lambda(u_1, ..., u_n)$ is in $GL_n(\mathbb{R})$ for each $(u_1, ..., u_n) \in \tilde{U}_\lambda$.
\end{itemize}
for each $\lambda \in \Lambda$.
\end{definizione}
\begin{remark}
The term "weakly nonlocal" comes from the work (Ref.~\cite{MaltNov}) of A.Ya. Maltsev and S.P. Novikov.
\end{remark}
\begin{remark}
In classical Poisson geometry there is an additional condition that has to be satisfied by Poisson brackets, namely the Leibniz formula:
\begin{equation*}
    \lbrace F  G, H \rbrace = F \lbrace G, H \rbrace  +  G \lbrace F , H \rbrace \quad \quad \quad \quad \forall F,G,H \in C^\infty(M)
\end{equation*}
By Remark \ref{noproduct} we see that there is no sense in requiring the validity of some analogous identity at the level of our functionals, as we don't have a well defined product between such objects (the theory of Hamiltonian PDEs arising from these infinite dimensional Poisson structures is not affected by the loss of this identity, see for example (Ref. \cite{Olver}) page 425).
\end{remark}

First of all, being the variational derivative of a WNL functional bounded, one obtains that the integrals in (\ref{poisson}) are convergent. This means that fixed $F,G \in \mathcal{W}(M,y)$ our bracket gives a well defined map
\begin{align*}
\lbrace F , G \rbrace : \bigsqcup_{\lambda \in \Lambda} \mathcal{S}^\lambda \rightarrow \mathbb{R}
\end{align*} 
Now we have to check which conditions on the elements ($g$, $\Gamma$, $w$) defining the bracket allow us to factor this map  through the projection induced by $\sim$.
In the case of these brackets, the following well known geometric characterization holds:
\begin{proposizione}
A family of maps of the form (\ref{poisson}) defines a map $\mathcal{W}(M,y) \times \mathcal{W}(M,y) \rightarrow \mathcal{W}(M,y)$ if and only if the families $\left\lbrace g_\lambda, \Gamma_\lambda, w_\lambda \right\rbrace_{\lambda \in \Lambda}$ define on M a (2,0) tensor field, a connection and a (1,1) tensor field respectively.
\end{proposizione}

The proof of this result is just a computation and is therefore omitted (see for example (Ref.~\cite{Ferapontov})).
This result gives a first hint for studying the dependence of these structures from the base point $y$:
\begin{corollario}\label{bijection}
Given $y \in M$ and consider the set $\mathcal{F}_y$ of functionals $\mathcal{W}(M,y)^2 \rightarrow \mathcal{W}(M,y)$ of local form (\ref{poisson}). The result above establishes the existence of a bijection $\mathcal{F}_y \rightarrow \mathcal{F}_z$ $\forall y,z \in M$, that correlates operators defined by the same tensor fields $g, w$ and connection $\Gamma$.
\end{corollario}
\section{FERAPONTOV'S THEOREM}

The next part of this work is devoted to showing how this precise choice of the functional spaces allows us to prove in a simple way this theorem due to Ferapontov \cite{Ferapontov}.
\begin{theorem}
\label{TH}
A bracket of the form (\ref{poisson}) defines a Poisson bracket if and only if its coefficients define on M a pseudometric $g$, its Levi Civita connection $\Gamma$ and the Gauss and Peterson-Codazzi-Mainardi equations hold.
\end{theorem}

Before giving a proof of this result we highlight one of its applications that allows to clarify what happens to the WNL-Poisson structures once we let the base point $y$ vary. The theorem implies that there is a canonical bijection between Poisson structures over different base points.
\begin{corollario}
	The same bijection defined in Corollary \ref{bijection} restricts to a bijection between Poisson brackets over $(M,y)$ and $(M,z)$.
\end{corollario}

In order to explicit the independence of Poisson structures form the choice of the base point, we can give the following interpretation of WNL-Poisson brackets over a manifold.
A WNL-Poisson bracket over (M,y) defines a family of WNL-Poisson brackets parameterized by their base point and defined by the same Riemannian objects. We will call such families \textit{WNL-Poisson brackets over M}.
In practice, such brackets are maps
\begin{align*}
	\lbrace \cdot , \cdot \rbrace : \bigsqcup_{y \in M} \mathcal{W}(M,y)^2 \rightarrow \bigsqcup_{y \in M} \mathcal{W}(M,y)
\end{align*}
defined by the commutativity of the following diagram for each $z \in M$ 
\begin{center}
	\begin{tikzpicture}
	\matrix (m) [matrix of math nodes, nodes in empty cells,row sep=3em,column sep=4em,minimum width=2em]
	{
		\bigsqcup_{y \in M} \mathcal{W}(M,y)^2 &  \bigsqcup_{y \in M} \mathcal{W}(M,y) \\
		\mathcal{W}(M,z)^2 &  \mathcal{W}(M,z) \\};
	\path[-stealth]
	(m-2-1) edge node [left] {$j_z$} (m-1-1)
	(m-1-1) edge node [above] {$\lbrace \cdot , \cdot \rbrace$} (m-1-2)
	(m-2-2) edge node [right] {$i_z$} (m-1-2)
	(m-2-1) edge node [below] {$\lbrace \cdot , \cdot \rbrace_z$} (m-2-2);
	\end{tikzpicture}
\end{center}
Where $i_z$ and $j_z$ are the inclusions in the disjoint unions and $\left \lbrace \lbrace \cdot, \cdot \rbrace_z \right \rbrace_{z \in M}$ is a family of WNL-Poisson brackets defined by the same pseudometric and Weingarten operator.\\
Let's now focus on theorem \ref{TH}.
The nature of this topic is local, so we will assume to be working on a fixed $U_\lambda$ without specifying it anymore. We will denote with $\Omega$ the open $\tilde{U}_\lambda$. 
To simplify the notation a bit we'll denote the derivation w.r.t. $x$ with $\prime$.
\begin{lemma}
\label{linearlemma}
Consider a bracket $\lbrace \cdot , \cdot \rbrace$ of the form (\ref{poisson}) and assume the skew-symmetry and the Jacobi identity hold for local functionals of the form
\begin{equation}
\label{linear}
F(u) := \int_\mathbb{R} \alpha_{i}(x) u^i(x) dx
\end{equation}
where $\alpha_{i} \in C^\infty_b(\mathbb{R})$. Then the bracket is a WNL Poisson bracket.
\end{lemma}
\begin{proof}
First of all, notice that applying the Euler-Lagrange formula to such an $F$ we get
\begin{equation*}
\dfrac{\delta F}{\delta u_i (x)} = \alpha_i(x)
\end{equation*}
Let $F, G \in \mathcal{W}(\Omega)$ and fix $w \in \mathcal{S}(\Omega)$.
If we define $\tilde{F}, \tilde{G} \in \mathcal{L}(\Omega)$ as
\begin{align*}
\tilde{F}(u) := \int_\mathbb{R} \dfrac{\delta F}{\delta w_i (x)} u^i(x) dx 
\quad \quad \quad 
\tilde{G}(u) := \int_\mathbb{R} \dfrac{\delta G}{\delta w_i (x)} u^i(x) dx
\end{align*}
we have that $\tilde{F}, \tilde{G}$ are of the form (\ref{linear}) and
\begin{align*}
\lbrace F , G \rbrace [w] &= \int_\mathbb{R} \dfrac{\delta F}{\delta w_i (x)} P^{i j}[w] \dfrac{\delta G}{\delta w_j (x)} dx\\
&= \int_\mathbb{R} \dfrac{\delta \tilde{F}}{\delta w_i(x)} P^{i j}[w] \dfrac{\delta \tilde{G}}{\delta w_j(x)} = \lbrace \tilde{F} , \tilde{G} \rbrace [w]
\end{align*}
the same argument holds for $\lbrace G , F \rbrace$, so
\begin{equation*}
\lbrace F , G \rbrace [w] = \lbrace \tilde{F} , \tilde{G} \rbrace [w] = -\lbrace \tilde{G} , \tilde{F} \rbrace [w] = -\lbrace G , F \rbrace [w]
\end{equation*}
Being $w$ arbitrary the thesis for the skew-symmetry follows.
For the Jacobi identity see (Ref.~\cite{Ferapontov}).
\end{proof}

We will use many times the following classical lemma, which we will state in a weak form.
\begin{lemma}[Variational Lemma]
Let $g \in C^{0}(\mathbb{R})$ and assume $\int_\mathbb{R} fg \, dx = 0$ for every $f \in C^{\infty}_0(\mathbb{R})$. Then $g = 0$.
\end{lemma}

The following result is an immediate application of what we have found above.
\begin{theorem}
A bracket of the form $(\ref{poisson})$ is skew-symmetric if and only if $g$ defines a pseudometric on M and the connection $\Gamma$ is compatible with $g$.
\end{theorem}

From now on we'll denote with $F$, $G$ and $H$ functionals of the form
\begin{align*}
&F(u) := \int_\mathbb{R} f_i(x)u^i(x) dx, \quad
G(u) := \int_\mathbb{R} g_j(x)u^j(x) dx, \quad\\
&H(u) := \int_\mathbb{R} h_l(x)u^l(x) dx
\end{align*}
where $f_i$, $g_j$ and $h_l$ belong to $C^\infty_b(\mathbb{R})$. Moreover, we will use the following notation:
\begin{align*}
\tilde{f} := d^{-1}\left( w_k^i u_x^k f_i\right), \quad
\tilde{g} := d^{-1}\left( w_k^j u_x^k g_j\right), \quad
\tilde{h} := d^{-1}\left( w_k^l u_x^k h_l\right)
\end{align*}
Thanks to our formula for the variational derivative of a $\tilde{\mathcal{W}}_n^1$ functional, a straightforward computation gives the following result for a skew-symmetric bracket of the form (\ref{poisson}):
\begin{align*}
\dfrac{\delta \lbrace F , G \rbrace}{\delta u^p} &= f^{\prime}_i g^{i s} \Gamma^j_{s p} g_j 
- f_i g^{s j} \Gamma^i_{s p} g^{\prime}_j + f_i u_x^k g_j R^{i j}_{p k}\\
& + f_i u_x^k g_j g^{s l} \left(\Gamma^i_{s p} \Gamma^j_{l k} - \Gamma^i_{s k} \Gamma^j_{l p} \right)
+ f_i u_x^k \left(w_k^i w_p^j - w_p^i w_k^j \right)g_j \\ 
&
+ \left[ f_i u_x^k \left( \dfrac{\partial w_k^i}{\partial u^p} - \dfrac{\partial w_p^i}{\partial u^k} \right) - f^{\prime}_i w_p^i \right] \tilde{g}\\
&- \left[ g_j u_x^k \left( \dfrac{\partial w_k^j}{\partial u^p} - \dfrac{\partial w_p^j}{\partial u^k} \right) - g^{\prime}_j w_p^j \right] \tilde{f}
\end{align*}
where $R^{i j}_{p k} := g^{i s} \left( \dfrac{\partial \Gamma^j_{s p}}{\partial u^k} - \dfrac{\partial \Gamma^j_{s k}}{\partial u^p} + \Gamma^l_{s p} \Gamma^j_{l k} - \Gamma^l_{s k} \Gamma^j_{l p} \right)$.\\
We are now ready to give a proof of Theorem \ref{TH}.
What remains to prove is that a skew-symmetric bracket satisfies the Jacobi identity iff
\begin{align}
\label{GPC:1}
&\Gamma_{s p}^{j} = \Gamma_{p s}^{j}\\
\label{GPC:2}
&R_{p k}^{i j} = w_{p}^{i} \, w_{k}^{j} - w_{k}^{i} \, w_{p}^{j}\\
\label{GPC:3}
&w_{p}^{i}\, g^{p l} = w_{p}^{l}\, g^{p i}\\
\label{GPC:4}
&\nabla_p w_k^{i} = \nabla_k w_p^{i}
\end{align}
where $R$ is the Riemann tensor of $(M,g)$ and $\nabla$ is its Levi Civita connection.
\begin{proof}
In this proof we will write $\partial_p$ instead of $\frac{\partial}{\partial u_p}$.
With the symbol $\circlearrowleft_{\alpha \beta \gamma}$ we'll denote the sum over the cycles of $S_3$ applied to $(\alpha, \beta, \gamma)$. So the Jacobi identity is written $\circlearrowleft_{F G H} \lbrace \lbrace F , G \rbrace , H \rbrace = 0$.
For brackets of the form (\ref{poisson}) this identity translates to
\begin{align*}
&\int_\mathbb{R} \circlearrowleft_{F G H} \left[ 
\dfrac{\delta \lbrace F , G \rbrace}{\delta u^p} \left(
g^{p l}h^{\prime}_l - g^{p s} \Gamma^l_{s k} u_x^k h_l + w_k^p u_x^k \tilde{h} \right) \right] dx =0
\end{align*}
Thanks to the previous calculation it's easy to compute the integrand above, which is
\begin{align*}
&- f_i \, g_j \, h_l u_x^{k}
a^{i j l}_k +
f^{\prime}_i g_j h_l u_x^k b^{i j l}_k
+
f_i g^{\prime}_j h_l u_x^k b^{j l i}_k
+
f_i \, g_j \, h_l^{\prime} u_x^{k} b^{l i j}_k\\
&+
\tilde{f} \, g_j \, h_l \, u_x^k c^{j l}_k 
+
f_i \, \tilde{g} \, h_l u_x^k c^{l i}_k 
+
f_i \, g_j \, \tilde{h} \, u_x^k c^{i j}_k\\
&+
f_i^{\prime} \, g_j^{\prime} \, h_l d^{i j l}
+
f_i^{\prime} \, g_j \, h_l^{\prime} d^{l i j}
+ 
f_i \, g_j^{\prime} \, h_l^{\prime} d^{j l i}\\
&+
f_i^{\prime} \, \tilde{g} \, h_l u_x^k e^{i l}_k
-
f_i^{\prime} \, g_j \, \tilde{h} u_x^k e^{i j}_k 
-
\tilde{f} \, g_j^{\prime} \, h_l u_x^k e^{j l}_k\\
&+
f_i \, g_j^{\prime} \, \tilde{h} u_x^k e^{j i}_k
+
\tilde{f} \, g_j \, h_l^{\prime} u_x^k e^{l j}_k
-
f_i \, \tilde{g} \, h_l^{\prime} u_x^k e^{l i}_k\\
&+
f_i^{\prime} \, g_j^{\prime} \, \tilde{h} m^{i j} 
+
f_i^{\prime} \, \tilde{g} \, h_l^{\prime}m^{l i}
+
\tilde{f} \, g_j^{\prime} \, h_l^{\prime} m^{j l}\\
\end{align*}
where
\begin{align*}
a^{i j l}_k &:= \left[\circlearrowleft_{ijl} \, g^{v s} \left( \Gamma_{v p}^{i}\, \Gamma_{s k}^{j} -  \Gamma_{v k}^{i}\, \Gamma_{s p}^{j}\right)\right. \\
& \left. + R_{p k}^{i j} - w_{p}^{i} \,w_{k}^{j} + w_{k}^{i}\, w_{p}^{j}\right] \,g^{p \alpha}\,\Gamma_{\alpha \beta}^{l}\,u_x^{\beta} \\
b^{i j l}_k &:= g^{i s} \, \Gamma_{s p}^{j} \,g^{p v}\,\Gamma_{v k}^{l} - g^{i s} \, \Gamma_{s p}^{l} \,g^{p v}\,\Gamma_{v k}^{j}\\
& + \left[ g^{v s} \left( \Gamma_{v p}^{j}\, \Gamma_{s k}^{l} -  \Gamma_{v k}^{j}\, \Gamma_{s p}^{l}\right) + R_{p k}^{j l} - w_{p}^{j} \,w_{k}^{l} + w_{k}^{j}\, w_{p}^{l}\right] \,g^{p i}\\
c^{j l}_k &:= \left(\partial_{p} w_\beta^{j} - \partial_\beta w_p^{j}\right)u_x^{\beta} \,g^{p \alpha}\,\Gamma_{\alpha k}^{l}\\
& + \left[ g^{v s} \left( \Gamma_{v p}^{j}\, \Gamma_{s \beta}^{l} -  \Gamma_{v \beta}^{j}\, \Gamma_{s p}^{l}\right) + R_{p \beta}^{j l} - w_{p}^{j} \,w_{\beta}^{l} + w_{\beta}^{j}\, w_{p}^{l}\right]u_x^{\beta}\,w_{k}^{p} \\
&- \left(\partial_{p} w_\beta^{l} - \partial_{\beta} w_p^{l}\right)u_x^{\beta} \,g^{p \alpha}\,\Gamma_{\alpha k}^{j}\\ 
d^{i j l} &:= g^{j s} \, \Gamma_{s p}^{l} \,g^{p i} - g^{i s} \, \Gamma_{s p}^{l} \,g^{p j}\\
e^{i j}_k &:= w_p^{i} \,g^{p \alpha}\,\Gamma_{\alpha k}^{j}- \left(\partial_{p} w_k^{j} - \partial_{k} w_p^{j}\right)\,g^{p i} - g^{i s} \, \Gamma_{s p}^{j}\,w_{k}^{p}\\
m^{i j} &:= w_p^{i} \,g^{p j} - w_p^{j} \,g^{p i}
\end{align*}
In the whole computation we have omitted the evaluation in $x$ and $u(x)$. Then the proof follows from the following two claims:
\begin{enumerate}
\item \emph{For a skew symmetric bracket of the form we consider, the Jacobi identity holds iff}
\begin{align}
\label{Jacobi}
b^{i j l}_k (z) = d^{i j l}(z) = e^{i j}_k (z) = m^{i j}(z) = 0 
\end{align}
\emph{for each $i, j, l, k \in \lbrace 1, ..., n\rbrace$ and for each $z \in \Omega$.}
\item \emph{The system (\ref{Jacobi}) is equivalent to the system (\ref{GPC:1}), (\ref{GPC:2}), (\ref{GPC:3}), (\ref{GPC:4})}.
\end{enumerate}
Let's start from the second one:
($\Rightarrow$) Using the symmetry of $g$ and renaming two indices we can write $0 = d^{i j l} = g^{i s} \left(\Gamma^l_{p s} - \Gamma^l_{s p}\right)g^{p j}$. In matricial form, defined $A(l) := \left(\Gamma^l_{p s} - \Gamma^l_{s p}\right)_{p,s = 1, ..., n} \in \mathbb{R}^{n \times n}$, this means $g A(l) g = 0$. By non degeneracy of the pseudometric $g$ it follows $A(l) = 0$ for each $l$, which is (\ref{GPC:1}). Now consider $b^{i j l}_k$; we can write it as
\begin{align*}
&g^{i s} \, g^{p v} \left( \Gamma_{s p}^{j} \,\Gamma_{v k}^{l} - \Gamma_{s p}^{l} \, \Gamma_{v k}^{j} \right) +  g^{v s} g^{p i} \left( \Gamma_{v p}^{j}\, \Gamma_{s k}^{l} -  \Gamma_{v k}^{j}\, \Gamma_{s p}^{l}\right)\\ & + g^{p i}\left(R_{p k}^{j l} - w_{p}^{j} \,w_{k}^{l} + w_{k}^{j}\, w_{p}^{l}\right)
\end{align*}
Consider the first two summands: renaming the indices so that $g^{i s} \, g^{p v}$ is a common factor we get that, using (\ref{GPC:1}), their sum is equal to
\begin{align*}
&g^{i s} \, g^{p v} \left(\Gamma_{v s}^{j}\, \Gamma_{p k}^{l}  - \Gamma_{s p}^{l} \, \Gamma_{v k}^{j}\right) = g^{i s} \, g^{p v} \, \Gamma_{v s}^{j}\, \Gamma_{p k}^{l} - g^{i s} \, g^{p v} \, \Gamma_{s p}^{l} \, \Gamma_{v k}^{j} \\
= &g^{i s} \, g^{p v} \, \Gamma_{v s}^{j}\, \Gamma_{p k}^{l} - g^{i s} \, g^{v p} \, \Gamma_{s v}^{l} \, \Gamma_{v p}^{j} = g^{i s} \, \Gamma_{s v}^{l} \, \Gamma_{v p}^{j} \left( g^{p v} - g^{v p} \right) = 0
\end{align*}
So $b^{i j l}_k = 0$ gives (\ref{GPC:2}) by the usual non degeneracy of $g$. Trivially $m^{i j} = 0$ is (\ref{GPC:3}). Using this last equation and renaming a couple of indices we get that
\begin{align*}
&0 = e^{i j}_k = g^{i p}\left(\partial_{p} w_k^{j} - \partial_{k} w_p^{j} + \Gamma_{p s}^{j}\, w_{k}^{s} - \Gamma_{s k}^{j} \,w_p^{s} \right) 
\\= &
g^{i p}\left[\left(\partial_{p} w_k^{j} + \Gamma_{p s}^{j}\, w_{k}^{s} - \Gamma_{k p}^{s} \,w_s^{j}\right) - \left(\partial_{k} w_p^{j}  + \Gamma_{s k}^{j} \,w_p^{s} - \Gamma_{p k}^{s} \,w_s^{j} \right)\right]
\\= &
g^{i p}\left(\nabla_p w_k^{j} - \nabla_k w_p^{j}\right)
\end{align*} 
which by non degeneracy of $g$ is (\ref{GPC:4}).\\
($\Leftarrow$ ) Is clear by looking at the definitions of $b$, $d$, $e$ and $m$.\\
Now let's consider our first claim. ($\Leftarrow$) This is the easiest implication of the two. It's just a matter of checking that $a^{i j l}_k$ and $c^{i j}_k$ are equal to zero for each $u$. But this consists in doing computations completely analogous to the ones above, so we omit them.\\
($\Rightarrow$)
If the Jacobi identity holds then the integral over $\mathbb{R}$ of the function in the previous page has to vanish for any choice of $f, g, h \in C^\infty_b(\mathbb{R}, \mathbb{R}^n)$ and $u \in \mathcal{S}(\Omega)$.
First of all, let's fix $i, j , l$ and consider $f, g, h$ having only one non zero component, respectively the $i$-th, $j$-th and $l$-th. So, we can erase the sum on those indices in the computation. In this part of the proof we assume that at least one of the functions $w^\alpha_\beta$ is non zero; namely, there exists $\bar{l},\bar{k} \in \lbrace 1, ..., n \rbrace$ and $\bar{z} \in \Omega$ such that $w^{\bar{l}}_{\bar{k}}(\bar{z}) \neq 0$.
The case where the $w$ are all zero is completely analogous, but simpler.
We can regroup the terms and write the Jacobi identity in the following form:
\begin{align*}
\int_\mathbb{R} (f\alpha + f^{\prime}\beta + \tilde{f}\gamma) dx = 0 \quad \quad \forall f \in C^\infty_b(\mathbb{R})
\end{align*} 
Restricting to functions $f \in C^\infty_0(\mathbb{R})$ we can integrate by parts getting
\begin{align*}
\int_\mathbb{R} f \left( \alpha - \beta^{\prime} - w^i_s u_x^s d^{-1}(\gamma)\right) dx = 0 \quad \quad \forall f \in C^\infty_0(\mathbb{R})
\end{align*} 
\begin{remark}
In the integration by parts of the third term we can neglect the boundary term by the property (\ref{secondProperty}) of the operator $d^{-1}$.
\end{remark}
Now we apply the variational lemma:
\begin{align}
\label{eq:3}
\alpha - \beta^\prime - w^i_s u_x^s d^{-1} (\gamma) = 0 \qquad \forall g,h \in C_b^\infty(\mathbb{R}) \quad \forall u \in \mathcal{S}(\Omega)
\end{align}
Explicitly this equation is
\begin{align*}
&- g \, h u_x^{k} a^{i j l}_k
- \left( g h u_x^k b^{i j l}_k \right)^\prime
+
g^{\prime}_j h u_x^k b^{j l i}_k
+
g \, h^{\prime} u_x^{k} b^{l i j}_k\\
&-
w^i_s u^s_x d^{-1}\left(g \, h \, u_x^k c^{j l}_k\right) 
+
\tilde{g} \, h u_x^k c^{l i}_k 
+
g \, \tilde{h} \, u_x^k c^{i j}_k\\
&-
\left( g^{\prime} \, h d^{i j l} \right)^\prime
-
\left(g \, h^{\prime} d^{l i j}\right)^\prime
+ 
g^{\prime} \, h^{\prime} d^{j l i}\\
&-
\left(\tilde{g} \, h u_x^k e^{i l}_k\right)^\prime
+
\left(g \, \tilde{h} u_x^k e^{i j}_k)\right)^\prime
+
w^i_s u^s_x d^{-1}\left( g^{\prime} \, h u_x^k e^{j l}_k\right)\\
&+
g^{\prime} \, \tilde{h} u_x^k e^{j i}_k
-
w^i_s u^s_x d^{-1}\left( g \, h^{\prime} u_x^k e^{l j}_k\right)
-
\tilde{g} \, h^{\prime} u_x^k e^{l i}_k\\
&-
\left( g^{\prime} \, \tilde{h} m^{i j} \right)^\prime
-
\left(  \tilde{g} \, h^{\prime}m^{l i} \right)^\prime
-
w^i_s u^s_x d^{-1}\left( g^{\prime} \, h^{\prime} m^{j l} \right) = 0
\end{align*}
\begin{remark}
In this part of the proof we will use the arbitrariness of $u \in \mathcal{S}(\Omega)$, $g, h \in C^\infty_b(\mathbb{R})$ to choose particular functions that, plugged into (\ref{eq:3}), give us relations that will imply the claim. In particular, we will use that given any bounded subset $B \subset \mathbb{R}$ and given any smooth function $f:B \rightarrow \mathbb{R}$ there exists $f^\star \in \mathcal{S}_1 \subset C^\infty_b(\mathbb{R})$ such that $f^\star_{\vert B} = f$. This follows by the existence of bump functions.
\end{remark}  From now on we consider a point $z$ on our $\Omega$. Consider $u \in \mathcal{S}(\Omega)$ such that
\begin{align*}
&u(0) = z \quad , \quad u_x^s(0) = 0 \quad , \quad u_{x x}^s(0) = 0\\
&u(1) = \bar{z} \quad , \quad u_x^s(1) = \delta^{s \bar{k}} \quad \quad \forall s \in \lbrace 1, ..., n \rbrace
\end{align*}
Plugging this $u$ in (\ref{eq:3}) and evaluating at $x=0$ we get
\begin{align*}
&g^{\prime \prime}\tilde{h} m^{i j} + g^\prime h m^{i j} + g h^\prime m^{l i} + \tilde{g}h^{\prime \prime} m^{l i} + g^{\prime \prime} h d^{i j l} + g^{\prime} h^\prime d^{i j l}\\ &+ g^{\prime} h^\prime d^{l i j} + g h^{\prime \prime} d^{l i j} - g^{\prime} h^\prime d^{j l i} = 0
\end{align*}
Now we choose $g$ such that $g(0) = g^\prime(0) = 0$ and $g^{\prime \prime}(0) \neq 0$. Then we have
\begin{align}\label{lasteq}
&g^{\prime \prime}\tilde{h} m^{i j} + \tilde{g}h^{\prime \prime} m^{l i} + g^{\prime \prime} h d^{i j l} = 0
\end{align}
What we have written until now holds for any choice of indices $i, j, l$. Now choose $l := \bar{l}$. Then we can construct $h$ such that $h(0) = h^\prime(0) = h^{\prime \prime}(0) = 0$ and $\tilde{h}(0) \neq 0$ through the use of bump functions. For this choice of $h$ our equation becomes $m^{i j}(z) = 0$ for each $i, j$. Then the equation (\ref{lasteq}) implies $d^{i j l} = 0$ for every $i, j, l$.
Now the equation (\ref{eq:3}) is considerably simplified. Let's fix $k$ and consider another $u \in \mathcal{S}(\Omega)$ such that
\begin{align*}
&u(0) = z \quad , \quad u_x^s(0) = 0 \quad , \quad u_{x x}^s(0) = \delta^{s k}\\
&u(1) = \bar{z} \quad , \quad u_x^s(1) = \delta^{s \bar{k}} \quad \quad \forall s \in \lbrace 1, ..., n \rbrace
\end{align*}
The equation becomes
\begin{align*}
& g \, h \, b^{i j l}_k
+
\tilde{g} \, h \, e^{i l}_k
-
g \, \tilde{h} \, e^{i j}_k = 0
\end{align*}
Fixing $l := \bar{l}$ we can choose $h$ as above, so we get $e^{i j}_k(z) = 0$ for each $i, j, k$ and hence $b^{i j l}_k = 0$ for every $i, j, l, k$.
\end{proof}

Let's now apply our theorem to see that a certain WNL-bracket is actually a WNL-Poisson bracket.
\begin{example}
Consider the bracket over $\mathbb{R}^2$ defined through (\ref{poisson}) by the following Riemannian objects:
\begin{itemize}
\item $g$ is diagonal defined by
\begin{align*}
	g^{1 1}(u,v) := - \alpha(u)(u-v)^2 \quad \quad\quad g^{2 2}(u,v) := \beta(v)(u-v)^2
\end{align*} 
where $\alpha(u) := c_1 + k + c_2u + c_3 u^2$ and $\beta(v) := c_1 + c_2v + c_3 v^2$.
\item $\Gamma_{i j}^k$ are defined as the Christoffel symbols of the Levi-Civita connection of $(\mathbb{R}^2,g)$, which are
\begin{align*}
	\Gamma_{11}^1 = \frac{1}{v-u} - \frac{\alpha^\prime(u)}{2\alpha(u)} \quad\quad&\quad\quad \Gamma_{22}^2 = \frac{1}{u-v} - \frac{\beta^\prime(v)}{2\beta(v)} \\
	\Gamma_{21}^1 = \Gamma_{12}^1 = \frac{1}{u-v}  \quad\quad&\quad\quad \Gamma_{21}^2 = \Gamma_{12}^2 = \frac{1}{v-u}\\
	\Gamma_{11}^2 = \frac{\beta(v)}{\alpha(u)(u-v)} \quad\quad & \quad\quad \Gamma_{22}^1 = \frac{\alpha(u)}{\beta(v)(u-v)}
\end{align*}
\item $w^i_j(u,v) := \delta^i_j \sqrt{k}$
\end{itemize}
where $k \in \mathbb{R}^+$ is fixed. In order to say that they define a WNL-Poisson bracket we have to check that these objects satisfy the Gauss and Peterson-Codazzi-Mainardi equations.
Clearly both terms in (\ref{GPC:4}) vanish being $w$ constant, so this equation is satisfied. To see that (\ref{GPC:3}) holds true, just compute:
\begin{align*}
	w^i_p g^{p l} = \sqrt{k} \delta^i_p g^{p l} = \sqrt{k} g^{i l} = \sqrt{k} g^{l i} = \sqrt{k} \delta^l_p g^{p i} = w^l_p g^{p i}
\end{align*}
For (\ref{GPC:2}), after some computations, we see that
\begin{align*}
	R^{i j}_{p k} = K \left( \delta^i_p \delta^j_k - \delta^i_k \delta^j_p \right) = w_{p}^{i}  w_{k}^{j} - w_{k}^{i}  w_{p}^{j}
\end{align*}
This shows that these are actually WNL-Poisson brackets over $\mathbb{R}^2$. The hypersurface of $\mathbb{R}^3$ associated to this bracket is a hypersurface of positive constant curvature $k$. These brackets are useful to give a Hamiltonian structure to the Chaplygin gas equations \cite{Mokhov}.
\end{example}
\subsection*{Acknowledgments}
I would like to show my deepest gratitude to Professor Paolo Lorenzoni for the constant support he provided throughout the writing of this paper.
\bibliographystyle{plain}
\bibliography{Paper}

\begin{thebibliography}{10}

\bibitem{KacDeSole}
{A. De Sole and V.G. Kac}.
\newblock {Non-local Poisson structures and applications to the theory of
  integrable systems}.
\newblock {\em {Japanese Journal of Mathematics}}, 8:233, 2013.

\bibitem{KrieglMich}
{A. Kriegl and P. W. Michor}.
\newblock {\em {The convenient setting of global analysis}}.
\newblock {American Mathematical Society}, 1997.

\bibitem{Stacey}
{A. Stacey}.
\newblock {The differential topology of loop spaces}.
\newblock {arXiv:math/0510097}, 2005.

\bibitem{MaltNov}
{A. Ya. Maltsev and S. P. Novikov}.
\newblock {On the local systems Hamiltonian in the weakly nonlocal Poisson
  brackets}.
\newblock {\em {Physica D: Nonlinear Phenomena}}, 156:53, 2001.

\bibitem{DubNov}
{B. A. Dubrovin and S. P. Novikov}.
\newblock {Hydrodynamics of weakly deformed soliton lattices. Differential
  geometry and Hamiltonian theory}.
\newblock {\em {Russian Mathematical Surveys}}, 44:35, 1989.

\bibitem{Gardner}
{C. S. Gardner}.
\newblock {Korteweg‐de Vries equation and generalizations. IV. The
  Korteweg‐de Vries equation as a Hamiltonian system}.
\newblock {\em {Journal of Mathematical Physics}}, 12:1548, 1971.

\bibitem{Ferapontov}
{E. V. Ferapontov}.
\newblock {Differential geometry of nonlocal Hamiltonian operators of
  hydrodynamic type}.
\newblock {\em {Functional Analysis and its Applications}}, 25:195, 1991.

\bibitem{ConstSav}
{G. M. Constantine and T. H. Savits}.
\newblock {A multivariate Faa di Bruno formula with applications}.
\newblock {\em {Transactions of the American Mathematical Society}}, 348:503,
  1996.

\bibitem{Lee}
{J. M. Lee}.
\newblock {\em {Introduction to smooth manifolds}}.
\newblock {Springer-Verlag}, 2012.

\bibitem{Milnor}
{J. W. Milnor}.
\newblock {\em {Topology from the differentiable viewpoint}}.
\newblock {Princeton University Press}, 1997.

\bibitem{Hormander}
{L. H\"{o}rmander}.
\newblock {\em {The analysis of linear partial differential operators III:
  Pseudo-differential operators}}.
\newblock {Springer-Verlag}, 2007.

\bibitem{Lorenzoni}
{M. Casati and P. Lorenzoni and R. Vitolo}.
\newblock {Three computational approaches to weakly nonlocal Poisson brackets}.
\newblock {arXiv:1903.08204} (preprint), 2019.

\bibitem{ReedSimon}
{M. Reed and B. Simon}.
\newblock {\em {Methods of modern mathematical physics: Functional analysis}}.
\newblock {Academic Press, Inc.}, 1980.

\bibitem{Mokhov}
{O. I. Mokhov}.
\newblock {Symplectic and Poisson structures on loop spaces of smooth
  manifolds, and integrable systems}.
\newblock {\em {Russian Mathematical Surveys}}, 53:515, 1991.

\bibitem{Olver}
{P. J. Olver}.
\newblock {\em {Applications of Lie groups to differential equations}}.
\newblock {Springer-Verlag}, 1986.

\bibitem{Lorenzoni2}
{P. Lorenzoni and R. Vitolo}.
\newblock {Weakly nonlocal Poisson brackets, Schouten brackets and
  supermanifolds}.
\newblock {arXiv:1909.07695} (preprint), 2019.

\bibitem{Hamilton}
{R. S. Hamilton}.
\newblock {The inverse function theorem of Nash and Moser}.
\newblock {\em {Bulletin of the American Mathematical Society}}, 7:65, 1982.

\bibitem{Lang}
{S. Lang}.
\newblock {\em {Differential manifolds}}.
\newblock {Springer-Verlag}, 1985.

\end{thebibliography}
\end{document}